\documentclass[11pt]{llncs}

\usepackage{a4wide}

\usepackage{showlabels}

\usepackage{amssymb}
\usepackage{bm}
\usepackage{latexsym}
\usepackage{eepic}
\usepackage{epic}
\usepackage{epsf}
\usepackage{graphics}

\spnewtheorem{Claim}[theorem]{Claim}{\bfseries}{\itshape}%{caption}{cap_font}{body_font}
\spnewtheorem{Lemma}[theorem]{Lemma}{\bfseries}{\itshape}%{caption}{cap_font}{body_font}
\spnewtheorem{Prop}[theorem]{Proposition}{\bfseries}{\itshape}%{caption}{cap_font}{body_font}
\spnewtheorem{Coro}[theorem]{Corollary}{\bfseries}{\itshape}%{caption}{cap_font}{body_font}
\spnewtheorem{Defi}[theorem]{Definition}{\bfseries}{\itshape}%{caption}{cap_font}{body_font}
\spnewtheorem{Remark}[theorem]{Remark}{\bfseries}{\itshape}%{caption}{cap_font}{body_font}

\newcommand{\bproof}{\noindent{\it Proof}}
\newcommand{\cproof}{\noindent{\it Proof of Claim}}
\newcommand{\eproof}{\hspace*{\fill}$\rule{2mm}{2mm}$~~~~~\bigskip}
\renewenvironment{proof}{\bproof. }{\eproof}
\newenvironment{psketch}{\bproof~{\it Sketch}. }{\eproof}

\newcommand{\mon}{\mbox{mon}}

\newcommand{\var}{\mbox{var}}
\newcommand{\Depth}{\mbox{Depth}}
\newcommand{\Width}{\mbox{Width}}
\newcommand{\NCDepth}{\mbox{ncDepth}}
\newcommand{\NCWidth}{\mbox{ncWidth}}

\newcommand{\NP}{\mbox{\rm NP}}

\newcommand{\poly}{\mbox{\rm poly}}

\newcommand{\NEXP}{\mbox{\rm NEXP}}
\newcommand{\NSUBEXP}{\mbox{\rm NSUBEXP}}

\newcommand{\Q}{\mathbb{Q}}
\newcommand{\Z}{\mathbb{Z}}
\newcommand{\real}{\mathbb{R}}
\newcommand{\nat}{\mathbb{N}}

\newcommand{\F}{\ensuremath{\mathbb{F}}}

\newcommand{\pder}[3]{\ensuremath{\partial^{#1}{#2} \over
  \partial{#3^{#1}}}}

\newcommand{\x}{\ensuremath{\overline{x}}}

\newcommand{\y}{\ensuremath{\overline{y}}}

\renewcommand{\angle}[1]{\langle #1\rangle}

\newcommand{\N}{\mbox{\rm N}}

\renewcommand{\P}{\mbox{\rm P}}

\renewcommand{\angle}[1]{\langle #1\rangle}

\title{On Lower Bounds for Constant Width Arithmetic Circuits}

\author{V.~Arvind, Pushkar S. Joglekar, Srikanth Srinivasan}

\institute{
  Institute of Mathematical Sciences\\
   C.I.T Campus,Chennai  600 113, India\\
   \email{\{arvind,pushkar,srikanth\}@imsc.res.in}% \date{}
}

\begin{document}

\maketitle

\begin{abstract}
  The motivation for this paper is to study the complexity of
  constant-width arithmetic circuits. Our main results are the
  following.
\begin{enumerate}
\item For every $k>1$, we provide an explicit polynomial that can be
	computed by a linear-sized monotone circuit of width $2k$ but has no
	subexponential-sized monotone circuit of width $k$. It follows, from
	the definition of the polynomial, that the constant-width and the
	constant-depth hierarchies of monotone arithmetic circuits are
	infinite, both in the commutative and the noncommutative settings.
\item We prove hardness-randomness tradeoffs for identity testing
  constant-width commutative circuits analogous to \cite{KI03,DSY}.
\end{enumerate}
\end{abstract}

\section{Introduction}

Using a rank argument, Nisan, in a seminal paper \cite{N91}, showed 
exponential size lower bounds for noncommutative formulas (and
noncommutative algebraic branching programs) that compute the
noncommutative permanent or determinant polynomials in the ring
$\F\angle{X}$, where $X=\{x_1,\cdots,x_n\}$ are noncommuting
variables.

By Ben-Or and Cleve's result \cite{BC}, we know that bounded-width
arithmetic circuits (both commutative and noncommutative) are at least
as powerful as formulas (indeed width three is sufficient). Can we
extend Nisan's lower bound arguments to prove size lower bounds for
\emph{noncommutative} bounded-width circuits?
Motivated by this question we make some simple motivating
observations in this section. We first recall some basic
definitions.

%\subsection*{Preliminaries}\label{intro}

%We begin with some basic definitions.

\begin{definition}\cite{N91,RS05}
  An \emph{Algebraic Branching Program} (ABP) over a field $\F$ and
  variables $x_1,x_2,\cdots,x_n$ is a \emph{layered} directed acyclic
  graph with one \emph{source} vertex of indegree zero and one
  \emph{sink} vertex of outdegree zero. Let the layers be numbered
  $0,1,\cdots,d$. Edges only go from layer $i$ to $i+1$ for each
  $i$. The source and sink are the unique layer $0$ and layer $d$
  vertices, respectively. Each edge in the ABP is labeled with a
  linear form over $\F$ in the input variables. The size of the ABP is
  the number of vertices. Each source to sink path in the ABP computes
  the product of the linear forms labeling the edges on the path, and
  the sum of these polynomials over all source to sink paths is the
  polynomial computed by the ABP.
\end{definition}

The scalars in an ABP can come from any field $\F$. If the input
variables $X=\{x_1,x_2,\cdots,x_n\}$ are noncommuting then the ABP (or
circuit) computes a polynomial in the free noncommutative ring
$\F\angle{X}$. If the variables are commuting then the polynomial
computed is in the ring $\F[X]$.

\begin{definition}\label{def_arith_ckt}
  An arithmetic circuit over $\F$ and variables $x_1,x_2,\cdots,x_n$
  is a directed acyclic graph with each node of indegree zero labeled
  by a variable or a scalar constant. Each internal node $g$ of the
  DAG is labeled by $+$ or $\times$ (i.e.\ it is a plus or multiply gate)
  and is of indegree two. A node of the DAG is designated as the
  output gate.  Each internal gate of the arithmetic circuit computes
  a polynomial (by adding or multiplying its input polynomials). The
  polynomial computed at the output gate is the polynomial computed by
  the circuit. The circuit is said to be \emph{layered} if its
  vertices are partitioned into vertex sets $V_1\cup V_2\cup\ldots\cup
  V_t$ such that $V_1$ consists only of leaves, and given any internal
  node $g$ in $V_i$ for $i>1$, the children of $g$ are either nodes
  from $V_1$ (consisting of constants or variables) or nodes from the
  set $V_{i-1}$. The size of a circuit is the number of nodes in it,
  and the width of a layered circuit is $\max_{i>1}|V_i|$.  An
  arithmetic circuit over the field $\real$ is \emph{monotone} if all
  the scalars used are nonnegative. Finally, a layered arithmetic
	circuit is \emph{staggered} if, in each layer $i$ with $i>1$, every
	node except possibly one is a product gate of the form $g = u \times 1$,
	for some gate $u$ from the previous layer.
\end{definition}

Note that the notion of bounded (i.e, constant) width staggered
circuits of width $w$ is identical to the notion of a straight-line
program with $w$ registers. The following lemma shows that staggered
circuits of width $w$ are comparable in power to width $w-1$ (not
necessarily staggered) arithmetic circuits. It holds in the
commutative and the noncommutative settings. We postpone the
proof of the lemma to the Appendix.

\begin{lemma}
	\label{lemma_staggered}
	Given any layered arithmetic circuit $C$ of width $w$ and size $s$
	computing a polynomial $p$, there is a staggered arithmetic circuit
	$C'$ of width at most $w+1$ and size $O(ws)$ computing the same
	polynomial.
\end{lemma}

A seminal result in the area of bounded width circuits is due to
Ben-Or and Cleve \cite{BC} where they show that size $s$ arithmetic
formulas computing a polynomial in $\F[X]$ (or in $\F\angle{X}$ in the
noncommutative case) can be evaluated by staggered arithmetic circuits of width
three and size $O(s^2n)$. Bounded width circuits have also been
studied under various restrictions in \cite{LMR07,MR08,JR09}.
However, they have not considered the question of proving explicit
lower bounds.

What is the power of arithmetic circuits of width $2$? It is easy to
see that the width-two circuit model is universal. We state this
(folklore) observation.

\begin{proposition}
  Any polynomial of degree $d$ with $s$ monomials in
  $\F[x_1,x_2,\cdots,x_n]$ (or in $\F\angle{x_1,\cdots,x_n}$) can be
  computed by a width two arithmetic circuit of size $O(d\cdot s)$.
  Furthermore, any monotone polynomial (i.e, with non-negative real
  coefficients) can be computed by a width two monotone circuit over
  $\real$ of size $O(d\cdot s)$.
\end{proposition}

\subsection*{Some Observations}

To motivate the study of constant-width circuits, we point out that,
for the problem of proving lower bounds for noncommutative bounded
width circuits, Nisan's rank argument is not useful.  For the
noncommutative ``palindromes'' polynomial
$P(x_0,x_1)=\sum_{w\in\{x_0,x_1\}^n} ww^R$, the communication matrix
$M_n(P)$ is of rank $2^n$ and hence any noncommutative ABP for it is
exponentially large \cite{N91}.  However, we can give an easy width-2
noncommutative arithmetic circuit for $P(x_0,x_1)$ of size $O(n)$.
Indeed, we can even ensure that each gate in this circuit is
\emph{homogeneous}.

\begin{proposition}
  The palindromes polynomial $P(x_0,x_1)$ has a width-2 noncommutative
  arithmetic circuit of size $O(n)$.
\end{proposition}

What then is a good candidate explicit polynomial that is not
computable by width-2 circuits of polynomial size? We believe that the
polynomial $P^\ell_k$ (of Section~\ref{mon}) for suitable $k$ is the
right candidate. A lower bound argument still eludes us. However, if
we consider \emph{monotone} constant-width circuits then even in the
commutative case we can show exponential size lower bounds for
monotone width-$k$ circuits computing $P^\ell_k$. Since $P^\ell_k$ is
computable by depth $2k$ arithmetic circuits (of unbounded fanin), it
follows that the constant-width and the constant-depth hierarchies of
monotone arithmetic circuits are infinite. We present these results in
Section~\ref{mon}.

\begin{remark}
  Regarding the separation of the constant-depth hierarchy of monotone
  circuits, we note that a separation has also been proved by Raz and
  Yehudayoff in \cite{RY09}; their lower bounds show a superpolynomial
  separation between the power of depth $k$ \emph{multilinear}
  circuits and depth $k+1$ monotone circuits for any $k$ (see
  \cite{RY09} for the definition and results regarding multilinear
  circuits). In contrast, our separation works only for monotone
  circuits, and only for infinitely many $k$. Nonetheless, we think
  that our separation is interesting because the separation we achieve
  is stronger. More precisely, the results of \cite{RY09} show a
  separation of the order of $2^{(\log s)^{1+\Omega(1/k)}}$ (that is,
  there is a polynomial that can be computed by circuits of depth
  $k+1$ and size $s$ but not by depth $k$ circuits of size $2^{(\log
    s)^{1+\Omega(1/k)}}$). On the other hand, our separation is at
  least as large as $2^{(\log s)^c}$ for any $c>0$ (see
  Section~\ref{mon} for the precise separation).
\end{remark}

A related question is the comparative power of noncommutative ABPs and
noncommutative formulas. Noncommutative formulas have polynomial size
noncommutative ABPs. However, $s^{O(\log s)}$ is the best known
formula size upper bound for noncommutative ABPs of size $s$. An
interesting question is whether we can prove a separation result
between noncommutative ABPs and formulas. We note that such a
separation in the \emph{monotone} case follows from an old result of
Snir \cite{S80}.

\begin{proposition}\hfill{~}\label{proposition2}
  Consider two noncommuting variables $\{x_0,x_1\}$. Let $L$ denote
  the set of all monomials of degree $2n$ with an equal number of
  $x_0$ and $x_1$, and consider the polynomial
  $E\in\Q\angle{x_0,x_1}$, where $E=\sum_{w\in L} w$. 
%	Then, $E$ has a
%	monotone homogeneous ABP of size $O(n^2)$, whereas any monotone
%	formula computing $E$ is of size $n^{\Omega(\lg n)}$.
\begin{itemize}
\item[1.] There is a monotone homogeneous ABP for $E$ of size $O(n^2)$.
\item[2.] Any monotone formula computing $E$ is of size $n^{\Omega(\lg n)}$.
\end{itemize}
\end{proposition}

\begin{proof}
  The first part is directly from a standard $O(n^2)$ size DFA that
  accepts precisely the set $L=\{w\in\{x_0,x_1\}^{2n}\mid w$ has an
  equal number of $x_0$'s and $x_1$'s$\}$. The second part follows
  from the fact that such a monotone formula would yield a commutative
  monotone formula for the symmetric polynomial of degree $n$ over the
  variables $y_1,y_2,\cdots,y_{2n}$: this is obtained by first
  observing that the formula must compute homogeneous polynomials at
  each gate.  Furthermore, we can label each gate (and each leaf) by a
  triple $(i,j,d)$ where $j-i+1=d$ is the degree of the homogeneous
  polynomial computed at this gate such that each monomial generated
  at this gate will occupy the positions from $i$ to $j$ in the output
  monomials containing it. Hence we have $x_0$'s at the leaf nodes
  labeled by triples $(i,i,1)$ for all $2n$ values of $i$. We replace
  the $x_0$'s labeled $(i,i,1)$ by $y_i$ and each $x_1$ by $1$. The
  resulting formula computes the symmetric polynomial as claimed.
  Snir in \cite{S80} has shown a tight $n^{\Omega(\log n)}$ lower
  bound for monotone formulas computing the symmetric polynomial of
  degree $n$ over the variables $y_1,y_2,\cdots,y_{2n}$.
\end{proof}

To illustrate again the power of constant width circuits, we note that
there is, surprisingly, a width-$2$ circuit for computing the
polynomial $E$.

\begin{proposition}
  There is a width-2 circuit of size $n^{O(1)}$ for computing $E$ if
  the field $\F$ has at least $cn^2$ distinct elements for some constant
  $c$.
\end{proposition}

\begin{psketch}
  This is based on the well-known Ben-Or trick \cite{B} for computing
  the symmetric polynomials in depth $3$. We consider the polynomial
	$g(x_0,x_1,z) = (x_0z^{2^{k+1}+1}+x_1z+1)^{2^{k+1}}$, where $2^{k-1}< n \leq 2^k$.
  ($z$ will eventually be a scalar from \F.) The coefficient
	of $z^{(2^{k+1}+1)n + n}$ in $g$ is precisely the polynomial $E$. Following Ben-Or's
  argument, the problem of recovering the polynomial $E$ can be
  reduced to solving a system of linear equations with an invertible
  coefficient matrix. Hence $E$ can be expressed as a sum
  $E=\sum_{i=1}^{2n}\beta_ig(x_0,x_1,z_i)$, where
  the $z_i$s are all distinct field elements. The terms
  $\beta_i g(x_0,x_1,z_i)$ can be evaluated with \emph{one}
	register using repeated squaring of $x_0z_i^{2^{k+1}+1}+x_1z_i+1$.  The
  second register is used as an accumulator to compute the sum of these
  terms.
\end{psketch}

These observations are additional motivation for the study of
constant-width arithmetic circuits. In Section~\ref{mon} we prove
lower bound results for monotone constant-width circuits. In
Section~\ref{pit} we explore the connection between lower bounds and
polynomial identity testing for constant-width commutative circuits
analogous to the work of Dvir et al \cite{DSY}.

\section{Monotone constant width circuits}\label{mon}

In this section we study \emph{monotone} constant-width
arithmetic circuits. We prove that they form an infinite hierarchy. As
a by-product, the separating polynomials that we construct yield the
consequence that constant-depth monotone arithmetic circuits too form
an infinite hierarchy. All our polynomials will be commutative, unless
we explicitly state otherwise.

For positive integers $k$ and $\ell$ we define a polynomial
$P^\ell_k$ on $\ell^{2k}$ variables as follows:
\begin{eqnarray*}
& P^\ell_1(x_1,x_2,\ldots,x_{\ell^2}) = \sum_{i=1}^\ell
\prod_{j=1}^\ell x_{(i-1)\ell + j}\\
& P^\ell_{k+1}(x_1,x_2,\ldots,x_{\ell^{2k+2}}) = \sum_{i=1}^\ell
\prod_{j=1}^\ell P^\ell_k(x_{(i-1)\ell^{2k+1} +(j-1)\ell^{2k} +
1},\ldots,x_{(i-1)\ell^{2k+1} + j\ell^{2k}})
\end{eqnarray*}

%Consider a depth $2k$ arithmetic formula $F^\ell_k$, all
%of whose gates are of fanin $\ell$. The output gate of $F^\ell_k$ is a
%$+$ gate and the formula has alternate layers of $+$ and $*$ gates.
%There are $n=\ell^{2k}$ leaves that are all labeled by the distinct
%variables $x_1,x_2,\cdots,x_n$ from left to right. We define
%$P^\ell_k$ to be the polynomial computed by $F^\ell_k$.

An easy inductive argument from the definition gives the following.

\begin{lemma}\label{number}
  The polynomial $P^\ell_k$ is homogeneous of degree $\ell^k$ on
  $\ell^{2k}$ variables and has $\ell^{\frac{\ell^k-1}{\ell-1}}$
  distinct monomials.
\end{lemma}

By definition, $P^\ell_k$ can be computed by a depth $2k$ monotone
formula of size $O(\ell^k)$. Furthermore, we can argue that the
polynomials $P^\ell_k$ are the ``hardest'' polynomials for
constant-depth circuits. We make this more precise in the following
observation.

\begin{proposition}\label{constant-depth-hard}
  Given a depth $k$ arithmetic circuit $C$ of size $s$, there is a
  projection reduction from $C$ to the polynomial $P^\ell_k$ where
  $\ell=O(s^{2k})$.
\end{proposition}

\begin{psketch}
  We sketch the easy argument. We can transform $C$ into a
  formula. Furthermore, we can make it a layered formula with $2k$
  alternating $+$ and $\times$ layers such that the output gate is a plus
  gate. This formula is of size at most $s^{2k}$. Clearly, a
  projection reduction (mapping variables to variables or constants)
  will transform $P^\ell_k$ to this formula, for $\ell=O(s^{2k})$.
\end{psketch}

It is easy to see the following from the fact that a monotone depth
$2k$ arithmetic circuit of size $s$ can be simulated by a monotone
width $2k$ circuit of size $O(s)$.

\begin{proposition}\label{width_2k_upper_bound}
  For any positive integers $\ell$ and $k$ there is a monotone circuit
	of width $2k$ and size $O(\ell^{2k})$ that computes $P^\ell_{2k}$.
\end{proposition}

We now state the main lower bound result. For each $k>0$ there is
$\ell_0\in\Z^+$ such that for all $\ell>\ell_0$ any width $k$ monotone
circuit for $P^\ell_k$ is of size $\Omega(2^\ell)$. We will prove this
result by induction on $k$. For the induction argument it is
convenient to make a stronger induction hypothesis.

For a polynomial $f\in\F[X]$, where $X=\{x_1,x_2,\cdots,x_n\}$ let
$\mon(f)=\{m\mid m$ is a nonzero monomial in $f\}$. I.e.\ $\mon(f)$
denotes the set of nonzero monomials in the polynomial $f$. Also, let
$\var(f)$ denote the set of variables occurring in the monomials in
$\mon(f)$. Similarly, for an arithmetic circuit $C$ we denote by
$\mon(C)$ and $\var(C)$ respectively the set of nonzero monomials and
variables occurring in the polynomial computed by $C$. 

We call a layered circuit $C$ \emph{minimal} if there is no smaller
circuit $C'$ of the same width s.t $\mon(C) = \mon(C')$.  It can be
seen that for any monotone circuit $C$, there is a minimal circuit
$C'$ of the same width s.t $\mon(C') = \mon(C)$ and has the following
properties.

\begin{itemize}
\item The only constants used in $C'$ are $0$ and $1$. Furthermore, no
  gate is ever multiplied by a constant.
\item By the minimality of $C'$ every node $g$ in $C'$ has a path to
  the output node of $C'$. Hence, given any node $g$ in $C'$ computing
  a polynomial $p$, there is a monomial $m$ such that $\mon(m\cdot
  p)\subseteq\mon(C')$. In particular, this implies that if $C'$
  computes a homogeneous multilinear polynomial, then $p$ must be a
  homogeneous multilinear polynomial.
\item If $C'$ computes a homogeneous multilinear polynomial of degree
  $d$, and if a node $g$ in layer $i$ also computes a polynomial $p$
  of degree $d$, then in layer $i+1$, there is a sum gate $g'$ such
  that $g$ is one of its children. Thus, the gate $g'$ computes a
  homogeneous multilinear polynomial $p'$ of degree $d$ such that
  $\mon(p)\subseteq\mon(p')$. In particular,
  $\mon(p)\subseteq\mon(C')$.
\end{itemize}

We call a minimal circuit satisfying the above a \emph{good} minimal
circuit.  We now show a useful property of minimal circuits $C$, which
applies to circuits satisfying $\mon(C)\subseteq P^\ell_k$, for all
$\ell,k\geq 1$.

%\begin{lemma}\label{factor}
%	Let $C$ be any good minimal circuit computing a polynomial
%	equivalent to $P^\ell_k$, for some $k\geq 1$ and $\ell\in\Z^+$. Then
%	the output gate for $C$ is a plus gate.
%\end{lemma}

\begin{lemma}\label{lemma_low_deg_prod}
  Let $f = \sum_{i=1}^\ell P_i$ be a homogeneous monotone polynomial
  of degree $d\geq 1$ with $\var(P_i)\cap\var(P_j) = \emptyset$ for
  all $i\neq j$. Given any good minimal circuit such that
  $\mon(C)\subseteq\mon(f)$, we have the following: if a gate $g$ in
  $C$ computes a polynomial $p$ of degree less than $d$, or a product
  of two such polynomials, then $\var(p)\subseteq\var(P_i)$ for a
  \emph{unique} $i$.
\end{lemma}

\begin{proof}
  For any polynomial $q\in\F[x_1,x_2,\cdots,x_n]$ we can define a
  bipartite graph $G(q)$ as follows: one partition of the vertex set
  is $\mon(q)$ and the other partition $\var(q)$. A pair $\{x,m\}$ is
  an undirected edge if the variable $x$ occurs in monomial $m$. It is
  clear that the graph $G(f)$ is just the disjoint union of all the
  $G(P_i)$.
	
  If the polynomial $p$ computed by gate $g$ is of degree $d'< d$,
  then, since $C$ is good, there is a monomial $m$ of degree $d'-d$
  such that $\mon(m\cdot p)\subseteq\mon(C)\subseteq\mon(f)$.  This
  implies that $G(m\cdot p)$ is a subgraph of $G(f)$. On the other
  hand, $G(m\cdot p)$ is clearly seen to be a connected graph.  This
  implies that, in fact, $G(m\cdot p)$ is a subgraph of $G(P_i)$ for
  some $i$ and hence, $\var(p)\subseteq\var(P_i)$ for a \emph{unique}
  $i$. This proves the lemma in this case.

  Similarly, if $p$ is a product of two polynomials of degree less
  than $d$, then $G(p)$ is a connected graph, and by the above
  reasoning, it must be the subgraph of some $G(P_i)$. Hence, the
  lemma follows.
\end{proof}

%\begin{proof}
%  By definition we can write $P^\ell_k=\sum_{i=1}^\ell Q_i$, where the
%  $Q_1,Q_2,\cdots,Q_\ell$ are all degree $\ell^k$ polynomials such
%  that $|\var(Q_i)|=\ell^{2k-1}$ and $\var(Q_i)\cap
%  \var(Q_j)=\emptyset$.  For any polynomial
%  $f\in\F[x_1,x_2,\cdots,x_n]$ we can define a bipartite graph as
%  follows: one partition of the vertex set is the set of nonzero
%  monomials of $f$ and the other partition is the set of variables.  A
%  pair $\{x,m\}$ is an undirected edge if the variable $x$ occurs in
%  monomial $m$.
%
%  Now, consider this bipartite graph $G$ corresponding to the
%  polynomial $P^\ell_k$. Clearly, the bipartite graph for $P^\ell_k$
%  is \emph{disconnected}: it is the union of the $\ell$ vertex
%  disjoint bipartite graphs corresponding to the polynomials $Q_i$.
%
%  Now, suppose that the output gate of the monotone circuit $C$ is a
%  multiply gate. Then we can write $P^\ell_k=P_1* P_2$, where $P_1$
%	and $P_2$ are both \emph{monotone} polynomials. Since $C$ is good,
%	$P_1$ and $P_2$ are both homogeneous polynomials of degree at least
%	$1$. Clearly, the bipartite graph defined by $P_1* P_2$ is
%	connected: any two nodes are connected by a path of length at most
%	four. This contradicts the fact that the graph of $P^\ell_k$ is
%	disconnected. Hence the output gate of $C$ is a plus gate.
%\end{proof}

We now state and prove a stronger lower bound statement. It shows
that $P^\ell_k$ is even hard to ``approximate'' by polynomial size
width-$k$ monotone circuits.

\begin{theorem}\label{lbound1}
  For each $k>0$ there is $\ell_0\in\Z^+$ such that for all
  $\ell>\ell_0$ and any width-$k$ monotone circuit $C$ such that
\[  
\mon(C)\subseteq \mon(P^\ell_k) \textrm{ and } |\mon(C)|\geq
\frac{|\mon(P^\ell_k)|}{2},
\]
the circuit $C$ is of size at least $\frac{2^\ell}{10}$.
\end{theorem}

\begin{proof}
  Let us fix some notation: given $i\in \Z^+$ and $j\in[w]$, we denote
  by $g_{i,j}$ the $j$th node in layer $i$ of $C$ and by $f_{i,j}$ the
  polynomial computed by $g_{i,j}$. Also, given a set of monomials
  $M$, we say that a circuit $C_1$ \emph{computes} $M$ if
  $\mon(C_1)\supseteq M$.
	
  Without loss of generality, we assume throughout that $C$ is a good
	minimal circuit.  The proof is by induction on $k$. The case $k=1$
	is distinct and easy to handle.  Thus, we consider as the induction
	base case the case $k=2$.  Consider a width two monotone circuit $C$
	such that $\mon(C)\subseteq\mon(P^\ell_2)$ and $|\mon(C)|\geq
  |\mon(P^\ell_2)|/2=\ell^{\ell+1}/2$. Let $f$ denote the polynomial
  computed by $C$. By Lemma~\ref{number} both $f$ and $P^\ell_2$ are
  homogeneous polynomials of degree $d=\ell^2$.
	
  We write the polynomial $P^\ell_2$ as $\sum_{i=1}^\ell P_i$, where
  $\var(P_i) = \{x_{(i-1)\ell^{3}+1},\ldots,x_{i\ell^{3}}\}$. Note
  that $\var(P_i)\cap \var(P_j)=\emptyset$ for $i\neq j$. Let
  $f=\sum_{i=1}^\ell P_i'$ where $\mon(P_i')\subseteq\mon(P_i)$ for
  each $i$.

  Since $C$ is good and $f$ is homogeneous, each gate of $C$
  computes only homogeneous polynomials. Moreover, since
  $\mon(C)\subseteq\mon(P^\ell_2)$ and
  $\var(P_i)\cap\var(P_j)=\emptyset$ for $i\neq j$,
  Lemma~\ref{lemma_low_deg_prod} implies that given any node $g$ in
  $C$ that computes a polynomial $p$ of degree less than $d$ or a
  product of such polynomials satisfies $\var(p)\subseteq\var(P_i)$
  for \emph{one} $i$.  Consider the lowest layer ($i_0$ say) when the
  circuit $C$ computes a degree $d$ monotone polynomial. W.l.o.g
  assume that $f_{i_0,1}$ is such a polynomial. We list some crucial
  properties satisfied by $g_{i_0,1}$ and $C$.
\begin{enumerate}

\item By the minimality of $i_0$, the node $g_{i_0,1}$ is a product
 gate computing the product of polynomials of degree less than $d$.
 Hence, $\var(f_{i_0,1})\subseteq\var(P_i)$ for exactly one $i$.
 W.l.o.g\ , we assume $i=1$. Since $\deg(f_{i_0},1) = d$ and $C$ is
 good, we in fact have $\mon(f_{i_0,1})\subseteq\mon(P_1)$.

\item Since $\deg(f_{i_0,1})=d$ and $C$ is good, we know that there is
  a node $g_{i_0+1,j_{i_0+1}}$ that is a sum gate with $g_{i_0,1}$ as
  child; $g_{i_0+1,j_{i_0+1}}$ computes a homogeneous polynomial of
  degree $d$ and $\mon(f_{i_0+1,j_{i_0+1}})\supseteq\mon(f_{i_0,1})$.
  Iterating this argument, we see that there must be a sequence of
  nodes $g_{i,j_i}$, for $i> i_0$ such that for each $i$, $g_{i,j_i}$
  is a sum gate with $g_{i-1,j_{i-1}}$ as child, such that 
  $\mon(f_{i_0,1})\subseteq\mon(f_{i_0+1,j_{i_0+1}})
  \subseteq\mon(f_{i_0+2,j_{i_0+2}})\ldots$, and each $f_{i,j_i}$ is
  a homogeneous polynomial of degree $d$.  We assume, w.l.o.g, that
  $j_i = 1$ for each $i> i_0$.

\end{enumerate}

By the choice of $i_0$, note that the node $g_{i_0,2}$ either computes
a polynomial of degree less than $d$ or computes a product of
polynomials of degree less than $d$.  Hence,
$\var(f_{i_0,2})\subseteq\var(P_i)$ for some $i$. If $i>1$, we assume
w.l.o.g.  that $\var(p)\subseteq\var(P_2)$. Let us consider the
circuit $C$ with the variables in $\var(P_1)\cup\var(P_2)$ set to $0$.
The polynomial computed by the new circuit $C'$ is now $f' =
f-P_1'-P_2'=\sum_{i=3}^\ell P'_i$. Let $q_{i,j}$ denote the new
polynomial computed by the node $g_{i,j}$.  Note that each $q_{i_0,j}$
is now a constant polynomial.

Consider the monotone circuit $C''$ obtained from $C'$ as follows: we
remove all the gates below layer $i_0$; the gate $g_{i_0,2}$ in layer
$i_0$ is replaced by a product gate $c\times1$, where $c$ is the constant
it computes in $C'$; from layer $i_0$ onwards, all nodes of the form
$g_{i,1}$ are removed; in any edge connecting nodes $g_{i,1}$ and
$g_{i+1,2}$, the node $g_{i,1}$ is replaced by the constant $0$.
Clearly, $C''$ is a width $1$ circuit.  For ease of notation, we will
refer to the nodes of $C''$ with the same names as the corresponding
nodes in $C'$. For any node $g_{i,2}$ in $C''$ ($i\geq i_0$), let
$q_{i,2}'$ be the polynomial it now computes. Crucially, we observe
the following from the above construction.
\begin{Claim}
  For each $i \geq i_0$,
  $\mon(q_{i,2}')\supseteq\mon(q_{i,2})\setminus\mon(q_{i,1})$.
\end{Claim}

%At this point, we make a technical claim which has an easy inductive
%proof: For any $i\geq i_0$, there is a width $1$ monotone circuit that
%computes the set $\mon(q_{i,2})\setminus\mon(q_{i,1})$. The proof is by
%induction on $i$. The base case $i=i_0$ is trivial since
%$q_{i_0,2}$ is a constant $c$, and thus the width $1$ circuit $g = c*1$
%will do. For the induction step, say we have a width $1$ circuit $C_1$
%computing the set $\mon(q_{i,2})\setminus \mon(q_{i,1})$. Let $g$
%denote the output gate of $C_1$. Say $g_{i+1,2}$ has children $u$ and
%$v$ and is labelled by $\circ\in\{+,*\}$. 
%\begin{itemize}
% \item If neither $u$ nor $v$ is an internal node, then we are easily done.
% \item If the only child that is internal is $g_{i,2}$, we proceed
%				 as follows. Say $g_{i+1,2} = g_{i,2}*x$ for some variable
%				 $x$. Then, we must have $\deg(q_{i,2}) < d$ and hence,
%				 in fact, $C_1$ computes the entire set $\mon(q_{i,2})$.
%				 Just adding a gate $g' = g*x$ to $C_1$ will ensure that the
%				 new width $1$ circuit continues to have the desired property.
%				 The other cases are similar.
% \item If the only child that is internal is $g_{i,1}$, then
%				 $g_{i+1,2}$ must be a product gate. Thus,
%				 $\mon(q_{i+1,2})\subseteq\mon(q_{i,1})\subseteq\mon(q_{i+1,2})$.
%				 Hence, we are done by simply throwing away $C_1$ and taking
%				 the width $1$ circuit $g' = 0$.
% \item If both internal nodes $g_{i,1}$ and $g_{i,2}$ are
%				 children, then $g_{i+1,1}$ must be a $+$ gate. We construct
%				 the new width $1$ circuit by adding a gate $g' = g+0$ to
%				 $C_1$.
%\end{itemize}

We now finish the proof of the base case. Define a sequence $i_1< i_2
< \ldots< i_t$ of layers as follows: for each $j\in[t]$, $i_j$ is the
least $i>i_{j-1}$ such that
$\mon(q_{i,1})\supsetneq\mon(q_{i_{j-1},1})$, and $\mon(q_{i_t,1}) =
\mon(f')$. Clearly, $t$ is at most the size of $C$. Note that it must
be the case that $q_{i_j,1} = q_{i_j-1,1} + q_{i_j-1,2}$. Hence, we
have $\mon(q_{i_j,1}) = \mon(q_{i_j-1,1})\cup\mon(q_{i_j-1,2}) =
\mon(q_{i_{j-1},1})\cup(\mon(q_{i_j-1,2})\setminus\mon(q_{i_j-1,1}))$.
By the above claim, the set
$\mon(q_{i_j-1,2})\setminus\mon(q_{i_j-1,1})$, which we will denote by
$S_j$, can be computed by a width-$1$ circuit. Thus, $\mon(f') = \mon(q_{i_t,1})
= \mon(q_{i_0,1})\cup\bigcup_{j=1}^tS_j$, where each $S_j$ can be
computed by a width-$1$ circuit. Since $q_{i_0,1}$ is the zero
polynomial, we have $\mon(f') = \bigcup_{j=1}^t S_j$.

Now, consider any width-$1$ monotone circuit computing a set
$S\subseteq P_2^\ell$. It is easy to see that the set $S$ computed
must have a very restricted form.

%By the above
%discussion, it follows that all the monomials of
%$P_3,P_4,\cdots,P_\ell$ have to be computed using register $R_2$ and
%register $R_1$ is only available as an accumulator for adding to it
%degree $d$ subpolynomials of the different $P_i$.
%
%Since the circuit is monotone, we cannot power the contents of $R_2$.
%Since the polynomials $P_i$ are all multilinear, $R_2$ can only be
%used to first compute a linear form and then multiply it with distinct
%variables. Hence we have the following claim.

\begin{Claim}
  The set $S$ is of the form $\mon(p)$ where $p = (\sum_{i\in X_1}
  x_i)\prod_{j\in X_2}x_j$, and $X_1\cap X_2=\emptyset$.
\end{Claim}

Clearly, as each set $S_j$ satisfies $S_j\subseteq \var(P_i')$ for
some $i$, it can have at most $\ell^3$ monomials.  Therefore, if the
monotone circuit $C$ is of overall size less than $2^\ell$ then it can
compute a polynomial of the form $P_1'+P_2'+f'$, where $f'$ has at
most $2^\ell\ell^3$ monomials.  Since $|\mon(P_i')|\leq |\mon(P_i)|=
\ell^\ell$ for each $i$, we have for suitably large $\ell$
\[
|\mon(C)|\leq 2\ell^\ell + 2^\ell\ell^3 < 3\ell^\ell <
\frac{\ell^{\ell+1}}{2}=\frac{|\mon(P^\ell_2)|}{2}
\]
and the base case follows.\\

\noindent{\bf The induction step.}

Consider any monotone circuit $\hat{C}$ of width $k-1$ such that
$\mon(\hat{C})\subseteq\mon(P^\ell_{k-1})$ and $|\mon(\hat{C})|\geq
|\mon(P^\ell_{k-1})|/2$. As induction hypothesis we assume that
$\hat{C}$ must be of size at least $2^\ell/10$.

Let $P^\ell_k=\sum_{i=1}^\ell P_i$, with $\var(P_i) =
\{x_{(i-1)\ell^{2k+1}+1},\ldots,x_{i\ell^{2k+1}}\}$ as in the base
case. By definition, the $\ell$ variable sets $\var(P_i)$ are mutually
disjoint and each $P_i$ has degree $d=\ell^k$. It is convenient to also write $P_i=\prod_{j=1}^\ell
Q_{ij}$, where each $Q_{ij}$ is of type $P^\ell_{k-1}$. We have
$\var(Q_{ij}) =
\{x_{(i-1)\ell^{2k+1}+(j-1)\ell^{2k}+1},\ldots,x_{(i-1)\ell^{2k+1}+j\ell^{2k}}\}$.

We start by considering any width $k-1$ circuit $\hat{C}$ of size less
than $2^\ell/10$ such that $\mon(\hat{C})\subseteq\mon(P^\ell_k)$. For
any $i\in[\ell]$, by fixing all the variables outside $\var(P_i)$ to
$0$, we obtain a width $k-1$ circuit $\hat{C}_i$ of the same size s.t
$\mon(\hat{C}_i)\subseteq\mon(P_i)$. Further, by setting all the
variables outside $\var(Q_{ij})$ to $1$ for some $j\in [\ell]$, we
obtain a circuit $\hat{C}_{ij}$ s.t
$\mon(\hat{C}_{ij})\subseteq\mon(Q_{ij})$. By the induction
hypothesis, we see that $|\mon(\hat{C}_{ij})|\leq |\mon(Q_{ij})|/2$.
Clearly
$\mon(\hat{C}_i)\subseteq\mon(\hat{C}_{i1})\times\mon(\hat{C}_{i2})\times\ldots\times
\mon(\hat{C}_{i\ell})$.  Therefore, $|\mon(\hat{C}_i)|\leq
\prod_j|\mon(\hat{C}_{ij})|\leq |\mon(P_i)|/2^{\ell}$. Finally, as
$\mon(\hat{C}) = \bigcup_i\mon(\hat{C}_i)$,
$|\mon(\hat{C})|\leq\sum_i|\mon(\hat{C}_i)|\leq
|\mon(P^\ell_k)|/2^\ell$. We have established the following claim.

\begin{Claim}\label{sizeclaim}
  For any width $k-1$ circuit $\hat{C}$ of size less than $2^\ell/10$
  such that $\mon(\hat{C})\subseteq\mon(P^\ell_k)$, we have
$|\mon(\hat{C})|\leq {|\mon(P^\ell_k)|\over 2^\ell}$.
\end{Claim}

For the induction step, consider any monotone width-$k$ circuit $C$
such that $\mon(C)\subseteq\mon(P^\ell_k)$ and of size at most
$2^\ell/10$. We will show that $|\mon(C)| <
|\mon(P^\ell_k)|/2$. W.l.o.g, we can assume that $C$ is a good minimal
circuit. Let $f$ denote the polynomial computed by $C$; we write
$f=\sum_{i=1}^\ell P_i'$, where $\mon(P_i')\subseteq\mon(P_i)$ for each $i$. 

As in the base case, let $i_0$ be the first layer where a polynomial
of degree $d$ is computed.  W.l.o.g.\ we can assume that $f_{i_0,1}$
is such a polynomial. By the minimality of $i_0$, the node $g_{i_0,1}$
must be a product node with children computing polynomials of degree
less than $d$. This implies, as in the base case, that
$\var(f_{i_0,1})\subseteq\var(P_i)$ for a unique $i$.  W.l.o.g.\ we
assume that $i=1$. As before, we can fix a sequence of
nodes $g_{i,j_i}$ for each $i>i_0$ such that $g_{i,j_i}$ is a sum gate
with $g_{i-1,j_{i-1}}$ as a child. It is easily seen that
$\mon(f_{i_0,1})\subseteq\mon(f_{i_0+1,j_{i_0+1}})\subseteq\mon(f_{i_0+2,j_{i_0+2}})\ldots$,
and each $f_{i,j_i}$ computes a homogeneous polynomial of
degree $d$. Renaming nodes if necessary, we assume $j_i=1$ for all
$i$.

Now consider $f_{i_0,j}$ for $j>1$. By the minimality of $i_0$, we see
that each $f_{i_0,j}$ is either a polynomial of degree less than $d$
or a product of two such polynomials. Hence,
$\var(f_{i_0,j})\subseteq\var(P_s)$ for some $s\in[\ell]$. Thus, there
is a set $S\subseteq [\ell]$ s.t $|S| = k' < k$ such that
$\bigcup_{j>1}\var(f_{i_0,j})\subseteq\bigcup_{s\in S}\var(P_s)$.
Without loss of generality, we assume that those $s\in S$ that are greater
than 1 are among $ \{2,3,\ldots,k\}$.

Consider the circuit $C'$ obtained when each of the variables in
$\bigcup_{s\in [k]}\var(P_s)$ is set to $0$. Let $q_{i,j}$ be the
polynomial computed by $g_{i,j}$ in $C'$. The polynomial computed by
$C'$ is just $f' = f - \sum_{s\in [k]}P_s'$. Note that $q_{i_0,j}$ is
now simply a constant for each $j$, and that the size of $C'$ is at
most the size of $C$ which by assumption is bounded by $2^\ell/10$.
Using this size bound we will argue that $C'$ cannot compute too many
monomials.

We now modify $C'$ as follows: we remove all the gates below layer
$i_0$; each gate $g_{i_0,j}$ with $j>1$ is replaced by a product gate
of the form $c\times 1$ where $c$ is the constant $g_{i_0,1}$ computes in
$C'$; from layer $i_0$ onwards, all nodes of the form $g_{i,1}$ are
removed; in any edge connecting nodes $g_{i,1}$ and $g_{i+1,j}$ for
$j>1$, the node $g_{i,1}$ is replaced by the constant $0$. Call this
new circuit $C''$. Clearly, $C''$ has size at most the size of $C$ and
width at most $k-1$. For ease of notation, we will refer to the nodes
of $C''$ with the same names as the corresponding nodes in $C'$. For
any node $g_{i,j}$ in $C''$ ($i\geq i_0$ and $j>1$), let $q_{i,j}'$ be
the polynomial it now computes. As in the base case, we observe the
following from the above construction.
\begin{Claim}
For each $i \geq i_0$ and each $j>1$,
$\mon(q_{i,j}')\supseteq\mon(q_{i,j})\setminus\mon(q_{i,1})$.
\end{Claim}

Using this, we show that the circuit $C'$ was essentially just using
the gates $g_{i,1}$ to store the sum of polynomials computed using
width $k-1$ circuits.

Construct a sequence of layers $i_1 < i_2 < \ldots < i_t$ in $C'$ as
follows: for each $j\in[t]$, $i_j$ is the least $i>i_{j-1}$ such that
$\mon(q_{i,1})\supsetneq\mon(q_{i_{j-1},1})$, and $\mon(q_{i_t,1}) =
\mon(f')$. Surely, $t$ is at most the size of $C'$. Now, fix any $i_j$
for $j\geq 1$. Clearly, it must be the case that $q_{i_j,1} =
q_{i_j-1,1} + q_{i_j-1,s}$ for some $s>1$; therefore, we have
$\mon(q_{i_j,1})\subseteq\mon(q_{i_j-1,1})\cup(\mon(q_{i_j-1,s})\setminus\mon(q_{i_j-1,1}))$.
Denote the set $\mon(q_{i_j-1,s})\setminus\mon(q_{i_j-1,1})$ by $S_j$.
Since the above holds for all $j$, and $\mon(q_{i_j-1,1}) =
\mon(q_{i_{j-1},1})$, we see that $\mon(f') =
\mon(q_{i_t,1})\subseteq\mon(q_{i_0,1})\cup\bigcup_jS_j = \bigcup_j
S_j$, since $q_{i_0,1}$ is the zero polynomial.

We will now analyze $|S_j|$ for each $j$. By the above claim, there is
a width $k-1$ circuit $C''$ of size at most the size of $C$ such that
$S_j\subseteq \mon(C'')\subseteq P^\ell_k$. If the size of $C$ (and
hence that of $C'$ and $C''$) is at most $2^\ell/10$, it follows from
Claim~\ref{sizeclaim} that $|S_j|\leq |\mon(P^\ell_k)|/2^\ell$. Hence,
we see that $|\mon(f')|\leq t|\mon(P^\ell_k)|/2^\ell$, which is at
most $|\mon(P^\ell_k)|/10$. But we know that the polynomial $f$
computed by the circuit $C$ is of the form $f' + \sum_{i\in [k]}P_i'$,
where $|\mon(P_i')| \leq |\mon(P_i)| = |\mon(P^\ell_k)|/\ell$.
Therefore,
\[ |\mon(f)| \leq \frac{k}{\ell}|\mon(P^\ell_k)| + |\mon(f')| \leq
|\mon(P^\ell_k)|\left(\frac{k}{\ell}+\frac{1}{10}\right) <
\frac{|\mon(P^\ell_k)|}{2}
\]
for large enough $\ell$. This proves the induction step.
\end{proof}

For $k\in\Z^+$ and $c>0$ let $\Depth_{k,c}$ and $\Width_{k,c}$ denote
the set of families $\{f_n\}_{n>0}$ of monotone polynomials
$f_n\in\real[x_1,x_2,\ldots,x_n]$ computed by $c\cdot n^c$-sized
monotone circuits of depth $k$ and width $k$ respectively. For
$k\in\Z^+$, let $\Depth_k=\bigcup_{c>0}\Depth_{k,c}$ and
$\Width_k=\bigcup_{c>0}\Width_{k,c}$. Thus, $\Depth_k$ and $\Width_k$
denote the set of families of monotone polynomials computed by
$\poly(n)$-sized monotone circuits of depth $k$ and width $k$
respectively. Note that, for each $k\in\Z^+$ we have
$\Depth_k\subseteq\Width_k$. Moreover, from the definition of
$P^\ell_k$, we see that the family $\{P^{\lfloor
  n^{1/2k}\rfloor}_k\}_n\in\Depth_{2k}$. Finally, in Theorem
\ref{lbound1} we have shown that the family $\{P^{\lfloor
  n^{1/2k}\rfloor}_k\}_n\notin\Width_k$, for constant $k$. Hence, we
have the following corollary of Theorem \ref{lbound1}.

\begin{corollary}
  For any fixed $k\in\Z^+$, $\Width_k\subsetneq\Width_{2k}$ and
  $\Depth_k\subsetneq\Depth_{2k}$.
\end{corollary}

Theorem \ref{lbound1} can also be used to give a separation between
the power of circuits of width (respectively, depth) $k$ and $k+1$ for
infinitely many $k$. We now state this separation. For any $k\in\nat$
and any function $f:\nat\rightarrow\nat$, let us denote by $f^k$ the
\emph{$k$-th iterate} of $f$, i.e the function $\underbrace{f\circ
f\circ \ldots \circ f}_{k\ \mathrm{ times }}$. Given non-decreasing
functions $f,g:\nat\rightarrow\nat$, call $f$ a \emph{sub $1/k$-th
iterate of $g$} if $f^k(n) < g(n)$, for large enough $n$ (closely
related notions have been defined in \cite{Sz61} and \cite{RR97}). It
can be verified that sub $1/k$-th iterates of exponential functions
can grow fairly quickly: for example, for any $\varepsilon>0$ and any
$k,c\in\nat$, the function $2^{(\log n)^c}$ is a sub $1/k$-th iterate
of $2^{n^\varepsilon}$.

We now state the precise separation that can be inferred from the
above theorem. For any $k,n\in\nat$ with $k\geq2$ and any polynomial
$p\in\real[x_1,x_2,\ldots,x_n]$, let $w_k(p)$ (resepctively $d_k(p)$)
denote the size of the smallest monotone width $k$ (respectively depth
$k$) circuit that computes $p$. 

\begin{corollary}
	There is an absolute constant $\alpha>0$ such that the following
	holds.  Fix any $k\in\nat$ where $k\geq 2$. Also, fix any
	non-decreasing function $f:\nat\rightarrow\nat$ that is a sub
	$1/k$-th iterate of $2^{\alpha n^{1/2k}}$. Then, for
	large enough $n$, there is a monotone polynomial
	$p\in\real[x_1,x_2,\ldots,x_n]$ such that for some $k',k''\in
	\{k,k+1,\ldots,2k-1\}$, $w_{k'}(p) \geq f(w_{k'+1}(p))$ and
	$d_{k''}(p) \geq f(d_{k''+1}(p))$.
\end{corollary}

\begin{proof}
	Let $p$ denote the monotone polynomial $P^{\lfloor
	n^{1/2k}\rfloor}_k\in\real[x_1,x_2,\ldots,x_n]$. Theorem
	\ref{lbound1} tells us that $w_k(p) =\Omega(2^{\lfloor
	n^{1/2k}\rfloor})$. To obtain a lower bound on $d_k(p)$, note that
	any polynomial computed by a circuit of size $s$ and depth $k$ can
	be computed by a width $k$ circuit of size $O(s^k)$; this tells us
	that $d_k(p) = 2^{\Omega(n^{1/2k})}$. Hence, there is some constant
	$\beta>0$ such that $\min\{w_k(p),d_k(p)\} \geq 2^{\beta n^{1/2k}}$,
	for large enough $n$.
	
	By definition, $p = P^{\lfloor n^{1/2k}\rfloor}_k$ has a depth
	$2k$ circuit of size $O(n)$, i.e $d_{2k}(p) = O(n)$. Proposition
	\ref{width_2k_upper_bound} tells us that $w_{2k}(p) = O(n)$ also.
	Hence, for some constant $\gamma>0$ and large enough $n$, we
	have $\max\{w_{2k}(p),d_{2k}(p)\}\leq \gamma n$. 

	The above statements imply that $w_k(p) \geq g(w_{2k}(p))$ and
	$d_k(p) \geq g(d_{2k}(p))$, where $g(n) = 2^{\alpha n^{1/2k}}$ for
	some constant $\alpha > 0$ and $n$ is large enough. Now, fix any
	non-decreasing function $f:\nat\rightarrow\nat$ that is a sub $1/k$-th
	iterate of $g$. We see that $w_k(p) \geq g(w_{2k}(p)) >
	f^k(w_{2k}(p))$ for large enough $n$; clearly, this implies that for some
	$k'\in\{k,k+1,\ldots,2k-1\}$, we must have $w_{k'}(p) \geq f(w_{k'+1}(p))$.
	Similarly, there is also a $k''\in\{k,k+1,\ldots,2k-1\}$ such that
	$d_{k''}(p) \geq f(d_{k''+1}(p))$.
\end{proof}

Similar corollaries hold for noncommutative circuits too. We define
the polynomial $P^\ell_k$ in exactly the same way in the
noncommutative setting. Note that any monotone bounded width
noncommutative circuit computing $P^\ell_k$ automatically gives us a
monotone commutative circuit of the same size and width computing the
commutative version of $P^\ell_k$. Hence, the lower bound of Theorem
\ref{lbound1} also holds for noncommutative width-$k$ circuits. For
$k\in\Z^+$, let $\NCDepth_k$ and $\NCWidth_k$ denote the set of
families of monotone polynomials $\{f_n\in\real\langle
x_1,x_2,\ldots,x_n\rangle\ |\ n\in\Z^+\}$ computed by $\poly(n)$-sized
monotone (noncommutative) circuits of depth $k$ and width $k$
respectively. Analogous to the commutative case, we obtain the
following.

\begin{corollary}
  For any fixed $k\in\Z^+$, $\NCWidth_k\subsetneq\NCWidth_{2k}$ and
  $\NCDepth_k\subsetneq\NCDepth_{2k}$.
\end{corollary}

And finally, we observe that the separations between width and depth
$k$ and $k+1$ that hold in the commutative monotone case also hold in
the noncommutative monotone case. Define, for any $k,n\in\nat$ with
$k\geq2$ and any polynomial $p\in\real\langle
x_1,x_2,\ldots,x_n\rangle$, let $ncw_k(p)$ (resepctively $ncd_k(p)$)
denote the size of the smallest monotone width $k$ (respectively depth
$k$) circuit that computes $p$. We have the following.

\begin{corollary}
	There is an absolute constant $\alpha>0$ such that the following
	holds.  Fix any $k\in\nat$ where $k\geq 2$. Also, fix any
	non-decreasing function $f:\nat\rightarrow\nat$ that is a sub
	$1/k$-th iterate of $2^{\alpha n^{1/2k}}$. Then, for
	large enough $n$, there is a monotone polynomial
	$p\in\real\langle x_1,x_2,\ldots,x_n\rangle$ such that for some $k',k''\in
	\{k,k+1,\ldots,2k-1\}$, $ncw_{k'}(p) \geq f(ncw_{k'+1}(p))$ and
	$ncd_{k''}(p) \geq f(ncd_{k''+1}(p))$.
\end{corollary}

\section{Identity testing for constant width circuits}\label{pit}

In this section we study polynomial identity testing for
constant-width commutative circuits. Impagliazzo and Kabanets
\cite{KI03} showed that derandomizing polynomial identity testing is
equivalent to proving arithmetic circuit lower bounds. Specifically,
assuming that there are explicit polynomials that require
superpolynomial size arithmetic circuits, they use these polynomials
in a Nisan-Wigderson type ``arithmetic'' pseudorandom generator that
can be used to derandomized polynomial identity testing.  This idea
was refined by Dvir et al \cite{DSY} to show that if there are
explicit polynomials that require superpolynomial size constant-depth
arithmetic circuits then polynomial identity testing for
constant-depth arithmetic circuits can be derandomized (the precise
statement involves the depth parameter explicitly \cite{DSY}).

In this section we prove a similar result showing that hardness for
constant-width arithmetic circuits yields a derandomization of
polynomial identity testing for constant-width circuits. We say that a
family of multilinear polynomials $\{P_n\}_{n>0}$ where
$P_n(\x)\in\F[x_1,\cdots,x_n]$ is \emph{explicit} if the coefficient
of each monomial $m$ of the polynomial $P_n$ can be computed in time
$2^{n^{O(1)}}$.

Recall the notion of a staggered arithmetic circuit (Definition
\ref{def_arith_ckt}).

\begin{lemma}\label{homcomp}
	Let $f\in\F[x_1,x_2,\cdots,x_n]$ of degree $m$ be computed by a
	staggered arithmetic circuit of size $s$ and width $w$. Then
	$H_i(f)$ (the $i^{th}$ homogeneous component of $f$) can be computed
	by a staggered circuit of size $\poly(s,m)$ and width $w+O(1)$, provided $\F$
	has at least $\deg(f)+1$ many elements.
\end{lemma}

\begin{proof}
Define a new polynomial $g(\x,z)\in\F[x_1,x_2,\cdots,x_n,z]$ as
$g(\x,z)=f(x_1z,x_2z,\cdots,x_nz)$.

We can write $f(x_1z,x_2z,\cdots,x_nz)=\sum_{i=0}^m H_i(f)z^i$ where
$m=\deg(f)$ and $H_i(f)$ is the $i^{th}$ homogeneous part of $f$.
Let $\{z_0,z_1,\cdots,z_m\}$ be $m+1$ distinct field elements.
Consider the matrix $M$ defined as
\[
M~=~\left( \begin{array}{lr}
       1~ z_0~ z^2_0~ & \cdots~ z^m_0 \\
       1~ z_1~ z^2_1~ & \cdots~ z^m_1 \\
       \cdots \cdots & \cdots \cdots \\
       1~ z_m~ z^2_m~ & \cdots~ z^m_m \\
       \end{array}   
\right).
\]
We have the system of equations
\[
M(H_0(f),H_1(f),\cdots,H_m(f))^T=(g(\x,z_0),g(\x,z_1),\cdots,g(\x,z_m))^T.
\]
Since $M$ is invertible, it follows that there are scalars
$a_{ij}\in\F$ such that $H_i(f)=\sum_{j=0}^m a_{ij}g(\x,z_j)$.

Since $f(\x)$ has a width $w$ circuit of size $s$, $g(\x,z)$ clearly
has a (staggered) circuit of width $w+O(1)$ of size $O(s)$. It follows easily from
the above equation for $H_i(f)$ that each $H_i(f)$ has a circuit of
width $w+O(1)$ and size $O(ms)$. 
\end{proof}

\begin{lemma}\label{pderiv}
  Let $P(x_1,x_2,\cdots,x_n,y)$ be a polynomial, over a sufficiently
	large field $\F$, computed by a width $w$ staggered circuit of size
	$s$. Suppose the maximum degree of $y$ in $P$ is $r$. Then for each
	$j$ the $j^{th}$ partial derivative $\pder{j}{P}{y}$ can be computed
	by a staggered circuit of width $w+O(1)$ and size $(rs)^{O(1)}$.
\end{lemma}

\begin{proof}
  Let $P(\x,y)=\sum_{i=0}^rC_i(\x)y^i$. As in Lemma~\ref{homcomp} each
  $C_i(\x)$ can be computed by a width $w+O(1)$ staggered circuit of size
  $O(rs)$. Clearly, for each $j$ the polynomial $\pder{j}{P}{y}$
can be written as
\[
{\pder{j}{P}{y}}={\sum_{i=j}^r a_{ij} C_i(\x)y^{i-j}},
\]
for $a_{ij}\in\F$, where $a_{ij}$ are field elements that depend
\emph{only} upon $j$. Therefore, we can easily give a staggered circuit of size
$O(r^2s)$ and width $w+O(1)$ for each polynomial $\pder{j}{P}{y}$.
\end{proof}

The following lemma is proved in \cite{DSY}. For any polynomial $g\in
\F[x_1,x_2,\cdots,x_n]$ let $H_{\leq k}(g)=\sum_{i=0}^kH_i(g)$. 

%It is proved by a finite taylor series expansion argument.

\begin{lemma}{\rm\cite[Lemma 3.2]{DSY}}\label{dsy}
  Let $P\in\F[x_1,x_2,\cdots,x_n,y]$ and $\deg_y(P)=r$. Suppose
  $f\in\F[x_1,x_2,\cdots,x_n]$ such that
  $P(\overline{x},f(\overline{x}))=0$ and
  ${\pder{}{P}{y}}(\overline{0},f(\overline{0}))$ is equal to $\xi\neq
  0$. Let $P(\overline{x},y)=\sum_{i=1}^r C_i(\overline{x})y^i$. Then
  for each $k\geq 0$ there is a polynomial
  $Q_k\in\F[y_0,y_1,\cdots,y_r]$ such that 
\[
H_{\leq k}(f)=H_{\leq k}(Q_k(C_0,C_1,\cdots,C_r)).
\]
\end{lemma}

Using the above lemmata we prove our first theorem.

\begin{theorem}\label{circ-solve}
  Let $P\in\F[x_1,x_2,\cdots,x_n,y]$ and $\deg_y(P)=r\geq 1$ such that
  $P$ has a staggered circuit of size $s$ and width $w$. Suppose that
  $P(\overline{x},f(\overline{x}))=0$ for some polynomial
  $f\in\F[x_1,x_2,\cdots,x_n]$ with $\deg(f)=m$. Then $f$ has a
  staggered circuit of size $\poly(s,(m+r)^r)$ and width $w+O(1)$ if
  $char(\F)>r$ and $\F$ is sufficiently large.
\end{theorem}

\begin{proof}
  First we argue that we can assume w.l.o.g., as in Dvir et al
  \cite{DSY}, that ${\pder{}{P}{y}}(\overline{0},f(\overline{0}))=\xi
  \neq 0$. If ${\pder{}{P}{y}}(x,f(x))\equiv 0$ we can replace $P$ by
  ${\pder{}{P}{y}}$. Since $char(\F)> r$ it is easy to see that there
  exists $j: 1\leq j\leq r$ such that
  ${\pder{j}{P}{y}}(x,f(x))\not\equiv 0$. Hence, we can assume
  ${\pder{}{P}{y}}(x,f(x))\not\equiv 0$. Therefore, there is an $a \in
  \F^n$ such that ${\pder{}{P}{y}}P(a,f(a))\neq 0$. We can assume that
  $a=0$ by appropriately shifting $P$ as in \cite{DSY}.  Let
\[
P(\overline{x},y)=\sum_{i=1}^r C_i(\overline{x})y^i.
\]
By Lemma~\ref{dsy} there is a polynomial $Q_k\in\F[y_0,\cdots,y_r]$
such that $H_{\leq k}(f)=H_{\leq k}(Q_k(C_0,C_1,\cdots,C_r))$ for each
$0\leq k\leq m$. Putting $k=m$ and letting $Q_m=Q$ we have
$f(\x)=H_{\leq m}(Q(C_0,C_1,\cdots,C_r))$.

Let $y^*=(C_0(0),\cdots,C_r(0))$ and $\deg(Q)=M$. Define
$I_M=\{(\alpha_0,\alpha_1,\cdots,\alpha_r)\mid \alpha_i\in\N, \sum
\alpha_i\leq M\}$. By expanding the polynomial $Q$ at the point $y^*$
we get $Q(\y)=\sum_{\overline{\alpha}\in
  I_M}Q_\alpha\prod_{i=0}^r(y_i-y^*_i)^{\alpha_i}$.

Thus, we can write
\[
f(\x)=H_{\leq m}[\sum_{\overline{\alpha}\in
  I_M}Q_\alpha\prod_{i=0}^r(C_i(\x)-C_i(0))^{\alpha_i}].
\]
As the constant term of $C_i(\x)-C_i(0)$ is zero, if we consider
$\prod_{i=1}^r (C_i(\x)-C_i(0))^{\alpha_i}$ for some
$\overline{\alpha}$ with $\sum_{i}\alpha_i>m$ then we will get
monomials of degree more than $m$ whose net contribution to $f(\x)$
must be zero. Hence we can write $f(\x)$ as
\[
f(\x)=H_{\leq m}[\sum_{\overline{\alpha}\in
  I_m}Q_\alpha\prod_{i=0}^r(C_i(\x)-C_i(0))^{\alpha_i}],
\]
where $I_m=\{(\alpha_0,\alpha_1,\cdots,\alpha_r)\mid \alpha_i\in\N,
\sum \alpha_i\leq m\}$. Clearly, $|I_m|\leq (m+r)^r$. Now, the
polynomial $\prod_{i=0}^r(y_i-y^*_i)^{\alpha_i}$ has a simple
$O(1)$-width circuit $C'$. We can compute
$\prod_{i=0}^r(C_i(\x)-C_i(0))^{\alpha_i}$ by plugging in the
staggered width $w+O(1)$ circuit for $C_i(\x)$ (as obtained in
Lemma~\ref{pderiv}) where $y_i$ is input to $C'$. Thus, we obtain a
circuit of width $w+O(1)$ for $\sum_{\overline{\alpha}\in I_m}
Q_\alpha\prod_{i=0}^r(C_i(\x)-C_i(0))^{\alpha_i}$ that is of size
polynomial in $s$ and $(m+r)^r$. By Lemma~\ref{homcomp} we can compute
its homogeneous components and their partial sums with constant
increase in width. Putting it together, it follows that $f(\x)$ can be
computed in width $w+O(1)$ of size polynomial in $s$ and $(m+r)^r$.
\end{proof}

We apply Theorem~\ref{circ-solve} to prove the main result of this
section.

\begin{theorem}\label{idtest} 
	There is a constant $c_1>0$ so that the following holds.  Suppose
	there is an explicit sequence of multilinear polynomials
	$\{P_m\}_{m>0}$ where $P_m(\x)\in\F[x_1,\cdots,x_m]$ and $P_m$
	cannot be computed by arithmetic circuits of width $w+c_1$ and size
  $2^{m^\epsilon}$, for constants $w\in\Z^+$ and $\epsilon>0$.
  Then, for any constant $c_2>0$, there is a deterministic $2^{(\log
    n)^{O(1)}}\cdot b^{O(1)}$ time algorithm that, when given as input
  a circuit $C$ of size $n^{O(1)}$ and width $w$ computing a
  polynomial $f(x_1,x_2,\ldots,x_n)$ of maximum coefficient size $b$,
  with each variable of individual degree at most $(\log n)^{c_2}$,
  checks if the polynomial computed by $C$ is identically zero,
  assuming that the field $\F$ is sufficiently large and $char(\F)>
  (\log n)^{c_2}$.
\end{theorem}

\begin{psketch}
	The overall construction is based on the Nisan-Wigderson
	construction as applied in Impagliazzo-Kabanets \cite{KI03} and Dvir
	et al \cite{DSY}. Hence it suffices to sketch the argument. 
\begin{enumerate}

\item Let $m=(\log n)^{c'(\epsilon,c_2)}$ and $\ell=(\log
  n)^{c''(\epsilon,c_2)}$ where $c''$ is suitably larger than $c'$.
\item Construct the Nisan-Wigderson design
  $S_1,\cdots,S_n\subset[\ell]$ such that $|S_i|=m$ for each $i$ and
  $|S_i\cap S_j|\leq \log n$.
\item Consider the polynomial
  $F(y_1,y_2,\cdots,y_\ell)=C(P_m(\y|S_1),P_m(\y|S_2),
  \cdots,P_m(\y|S_n))$.  For any input $\y\in\F^\ell$ we can evaluate
  $F$ by evaluating $P_m(\y|S_i)$ for each $i$ and then evaluating $C$
  on the resulting values. Since the $P_m$ are explicit polynomials
  and $|S_i|$ has polylog$(n)$ size we can evaluate $P_m$ in time
  $2^{(\log n)^{O(1)}}$.

\item We test if $F(\y)\equiv 0$ using a brute-force algorithm based
  on the Schwartz-Zippel lemma. Consider a finite set $S\subseteq \F$,
  such that $|S|$ is more than $\deg(F)$. Check if
  $F(\overline{a})\equiv 0$ for all $\overline{a} \in S^\ell$ in time
  $n^{O(\ell)}$.  If all the tests returned zero then return $C \equiv
  0$ otherwise $C \not\equiv 0$.
\end{enumerate}

The proof of correctness is exactly as in \cite{KI03,DSY}. Assuming
the algorithm fails, after hybridization and fixing variables in $C$,
we get a nonzero polynomial $F_2$ of the form 
\[
F_2(\y|S_{i+1},x_{i+1}) = F_1(P_m(\y|S_1\cap S_{i+1}),P_m(\y|S_2\cap
S_{i+1}),\cdots,P_m(\y|S_i\cap S_{i+1}),x_{i+1}).
\]
where $F_1(x_1,x_2,\ldots,x_{i+1})$ can be computed by a width
$w$ circuit of size $\poly(n)$ and
$F_2(\y|S_{i+1},P_m(\y|S_{i+1}))\equiv 0$. Note that the
multilinear polynomials $P_m(\y|S_j\cap S_{i+1})$ depend only on $\log
n$ variables and hence, they can be computed using brute force
width-$2$ staggered circuits of size $O(n\log n)$. Also, by
Lemma~\ref{lemma_staggered}, we know that $F_1$ can be computed
by a staggered circuit of size $\poly(n)$ and width at most
$w+1$. Putting the above circuits together, it is easy to see that
$F_2$ can be computed by a staggered circuit $C'$ of
size at most $\poly(n).n\log n = \poly(n)$ and width $w+O(1)$.  Now,
by applying Theorem~\ref{circ-solve} to $C'$, we get a circuit of
width $w+O(1)$ to compute $P_m$ contradicting the hardness assumption.
\end{psketch}

Finally, we observe that the following analogue of \cite[Theorem
4.1]{KI03} holds for bounded width circuits. The proof is in the
appendix.

\begin{proposition}
	\label{prop_ki_lbound}
One of the following three statements is false.
\begin{enumerate}
\item $\NEXP \subseteq \P/\poly$.
\item The Permanent polynomial is computable by polynomial size width
  $w$ arithmetic circuit over $\Q$, where $w$ is a constant.
\item The identity testing problem for bounded width arithmetic
  circuits over $\Q$ is in $\NSUBEXP$.
\end{enumerate}
\end{proposition}

\subsection*{Acknowledgements}
We would like to thank Amir Yehudayoff for pointing out the separation
in \cite{RY09} and for many valuable comments.

\newpage

\begin{center}
{\bf\Large Appendix}
\end{center}

\noindent
{\bf Proof Sketch of Lemma \ref{lemma_staggered}}

The circuit $C'$ is constructed by
showing how to compute, for $i\geq 1$, the polynomials computed in 
layer $i+1$ of $C$ from the polynomials computed in the $i$th layer in
$C$ in a staggered fashion, using at most $w$ layers of width at most
$w+1$. Equivalently, it amounts to designing a straight-line program
with $w+1$ registers such that: initially, $w$ of the registers
contain the polynomials computed in the $w$ nodes of the $i^{th}$
layer. In the end, $w$ of the $w+1$ registers will contain the
polynomials computed at the $i+1^{st}$ layer of $C$. Note that this is
trivial for $i = 2$ since all nodes in layer $2$ have only leaves
as children. For some $i>1$, let the $U$ denote the nodes of $C$ in
layer $i$ and $V$ the nodes of $C$ in layer $i+1$.

We define an undirected multigraph $G$ corresponding to layers $i$
and $i+1$ as follows: its vertex set $V(G)$ is $U$. For each gate
$v\in V$ in circuit $C$ that takes inputs $u_1,u_2\in U$ we include
the edge $\{u_1,u_2\}$ in $E(G)$. Notice that if $u_1=u_2$ we add a
self-loop to $E(G)$. Furthermore, if $v\in V$ takes one input as a
$u\in U$ and the other inputs is a constant or a variable, then too we
add a self-loop at vertex $u$. Finally, if both inputs to $v$ are
constants and/or variables, there is no edge in $G$ corresponding to
$v$.  We note some properties of this graph $G$.

\begin{enumerate}
\item We have $|V(G)|\leq w$ and $|E(G)|+|V'|\leq w$, where $V'$ is
	the set of those nodes in $V$ that take only constants and/or
	variables as input.

\item Each vertex $u\in V(G)$ corresponds to a polynomial $p_u$
  computed at $u$ in the $i^{th}$ layer. Each edge $e\in E(G)$ is
  defined by some $v\in V$ and it corresponds to the polynomial $q_e$
  computed at $v$. In order to compute the polynomial corresponding to
  $e$ we need the polynomials corresponding to its end points.
\end{enumerate}

We have $w+1$ registers, $w$ of which contain the polynomials $p_u,
u\in U$. Our goal is to compute the polynomials $q_e, e\in E(G)$ using
these registers. Using the graph structure of $G$, we will give an
ordering of the edges $E(G)$. If we compute the polynomials $q_e$ in
that order then for every $q_e$ computed we will have a free register
to store $q_e$ (when we do not need a polynomial $p_u$ for further
computation, we can free the register containing $p_u$).

Thus, what we want to do is compute an ordering of the edges
$E(G)$\footnote{We can blur the distinction between vertices and edges
  and the polynomials they represent.} from the vertex set $V(G)$. 

We pick edges from $E(G)$ one by one. When $e\in E(G)$ is picked, we
delete $e$ from the graph and store $q_e$ in a free register.
Crucially, note that when a vertex $u\in V(G)$ becomes isolated in
this process the polynomial $p_u$ is not required for further
computation and the register containing $p_u$ is freed. Thus, at any
point of time in this edge-deletion procedure, the number of registers
required is equal to the sum of the number of edges removed from $G$
and the number of \emph{non-isolated} vertices left in $G$.

The edge picking procedure works as follows. We break $G$ into its
connected components $G_1\cup G_2\cup\ldots\cup G_s\cup
G_{s+1}\cup\ldots\cup G_{s+t}$, where $G_1,G_2,\ldots,G_s$ are the
\emph{acyclic} components and $G_{s+1},\ldots,G_{s+t}$ have cycles. We
first compute the edges of $G_1$, and then those of $G_2$, and so on.
At the end, we compute the polynomials corresponding to the nodes in
$V'$.

Each connected component $G_i$ is processed as follows: if there is an
edge $e$ in $G_i$ that is not a cut edge, we pick the edge $e$ and
delete it from the graph; otherwise, since every edge of $G_i$ is a
cut edge, $G_i$ must be a tree, and in this case, we remove any edge
$e$ that is incident to a degree-$1$ vertex. Proceeding thus, we
maintain the invariant that at all points, all but one of the
components of $G_i$ are isolated vertices. We can use this to show
that the number of registers required at any point in the computation
of $q_e$ for $e\in E(G_i)$ is at most $|E(G_i)|+1$ (in particular, if
$G_i$ is acyclic this is at most $|V(G_i)|$).

Putting it all together, we can also show that the maximum number of
nodes used in computing the edges of $G$ is bounded by
$\max\{|V(G)|,|E(G)|+1,|E(G)|+|V'|\}\leq w+1$. Moreover, since at
each step the polynomial of some node $v\in V$ is computed, the total
number of steps in the straight-line program is at most $w$. This
proves the lemma.  
\vspace*{10pt}

\noindent
{\bf Proof of Proposition \ref{prop_ki_lbound}}

The proof follows the same lines as that of \cite[Theorem
4.1]{KI03}. A similar result for \emph{bounded-depth} circuits is
noted in \cite[Section 5]{DSY}. We give a brief proof sketch. Assume
to the contrary that all three statements hold. Following the proof in
\cite{KI03}, $\NEXP$ will collapse to $\NP^{\textrm{Perm}}$. Hence, it
suffices to show $\P^{\textrm{Perm}} \subseteq \NSUBEXP$ to derive a
contradiction (to the nondeterministic time hierarchy theorem). The
language $\textrm{Perm}$ consists of all tuples $(M,v)$, where $M$ is
an integer matrix and $v$ is the binary encoding of
$\textrm{Perm}(M)$. The NSUBEXP machine will guess a polynomial size,
width-$w$ circuit $C$ for the $n\times n$ Permanent polynomial over
variables $\{x_{ij}| 1\leq i,j \leq n\}$. Next, we want to check whether
$C$ indeed computes the Permanent polynomial. We can easily obtain a
width-$w$ polynomial-sized circuit $C_k$ that computes the permanent
of $k \times k$ matrix over variables $\{x_{ij}| 1\leq i,j \leq k\}$ from
circuit $C$. Next we check whether $B_1=C_1(x)-x \equiv 0$. For $n\geq
k>1$ check that $B_k= C_k(X^{(k)})- \sum_{i=1}^k
x_{1,i}C_{k-1}(X_i^{(k)}) \equiv 0$, where $X^{(k)}=(x_{i,j})_{i,k \in
  [k]}$ is the $k \times k$ matrix and $X_i^{(k)}$ is a minor obtained
by deleting first row and $i^{th}$ column of $X^{(k)}$. It follows
that if all the $B_i$'s are identically zero polynomials then $C$
computes the Permanent polynomial. Since $C_k$ has a width-$w$
polynomial size circuit it follows that $B_k$ can be computed by a
polynomial-size width $w+O(1)$ circuit. We can now use the assumed
deterministic subexponential time algorithm for identity testing of
bounded width circuits to check whether each $B_k$ is identically
zero. Putting it together, we have $\P^{\textrm{Perm}} \subseteq
\NSUBEXP$.

\end{document}